\documentclass[11pt,a4paper]{amsart}
\usepackage{amssymb,amsmath}
\usepackage{mathrsfs}
\usepackage{newtxtext}
\usepackage[alphabetic]{amsrefs}
\usepackage{graphicx,color}
\usepackage[utf8]{inputenc}
\usepackage{amsthm}
\usepackage{latexsym}
\usepackage{slashed}
\usepackage[all]{xy}
\usepackage{hyperref}
\usepackage[mathscr]{eucal}
\usepackage{tikz}
\usetikzlibrary{calc,decorations.markings}  
\usepackage{mathtools}
\usepackage{accents}
\usepackage{amsaddr}
\usepackage{cleveref}

\def\theequation{\thesection.\arabic{equation}}
\makeatletter
\@addtoreset{equation}{section}
\makeatother

\makeatletter
\newcommand{\eqnum}{\refstepcounter{equation}\textup{\tagform@{\theequation}}}
\makeatother

\newcounter{copy}
\makeatletter
\renewcommand{\thecopy}{\ifnum0=\c@section\arabic{copy}\else\thesection.\arabic{copy}'\fi}
\makeatother

\BibSpec{book}{%
    +{}  {\PrintPrimary}                {transition}
    +{.} { \textit}                     {title}
    +{.} { }                            {part}
    +{:} { \textit}                     {subtitle}
    +{,} { \PrintEdition}               {edition}
    +{}  { \PrintEditorsB}              {editor}
    +{,} { \PrintTranslatorsC}          {translator}
    +{,} { \PrintContributions}         {contribution}
    +{,} { }                            {series}
    +{,} { \voltext}                    {volume}
    +{,} { }                            {publisher}
    +{,} { }                            {organization}
    +{,} { }                            {address}
    +{,} { }                            {status}
    +{,} { \PrintDOI}                   {doi}
    +{,} { \PrintISBNs}                 {isbn}
    +{}  { \parenthesize}               {language}
    +{}  { \PrintTranslation}           {translation}
    +{;} { \PrintReprint}               {reprint}
    +{,} { \PrintDate}                  {date}
    +{.} { }                            {note}
    +{.} {}                             {transition}
}
\BibSpec{article}{%
    +{}  {\PrintAuthors}                {author}
    +{,} { \textit}                     {title}
    +{.} { }                            {part}
    +{:} { \textit}                     {subtitle}
    +{,} { \PrintContributions}         {contribution}
    +{.} { \PrintPartials}              {partial}
    +{,} { }                            {journal}
    +{}  { \textbf}                     {volume}
    +{}  { \PrintDatePV}                {date}
    +{,} { \issuetext}                  {number}
    +{,} { \eprintpages}                {pages}
    +{,} { }                            {status}
    +{,} { \PrintDOI}                   {doi}
    +{}  { \parenthesize}               {language}
    +{}  { \PrintTranslation}           {translation}
    +{;} { \PrintReprint}               {reprint}
    +{.} { }                            {note}
    +{.} {}                             {transition}
}
\BibSpec{collection.article}{%
    +{}  {\PrintAuthors}                {author}
    +{,} { \textit}                     {title}
    +{.} { }                            {part}
    +{:} { \textit}                     {subtitle}
    +{,} { \PrintContributions}         {contribution}
    +{,} { \PrintConference}            {conference}
    +{}  {\PrintBook}                   {book}
    +{,} { }                            {booktitle}
    +{,} { \PrintDateB}                 {date}
    +{,} { pp.~}                        {pages}
    +{,} { }                            {publisher}
    +{,} { }                            {organization}
    +{,} { }                            {address}
    +{,} { }                            {status}
    +{,} { \PrintDOI}                   {doi}
    +{,} { \eprint}        {eprint}
    +{}  { \parenthesize}               {language}
    +{}  { \PrintTranslation}           {translation}
    +{;} { \PrintReprint}               {reprint}
    +{.} { }                            {note}
    +{.} {}                             {transition}
}

\BibSpec{misc}{
  +{}{\PrintAuthors}  {author}
  +{,}{ \textit}     {title}
  +{}{ (}             {date}
  +{),}{ }             {note}
  +{.}{}              {transition}
}

\usepackage{cleveref}
\theoremstyle{definition}
\newtheorem{defn}[equation]{Definition}

\newtheorem{para}[equation]{}
\theoremstyle{plain}
\newtheorem{thm}[equation]{Theorem}

\newtheorem{lem}[equation]{Lemma}

\theoremstyle{remark}
\newtheorem{rmk}[equation]{Remark}

\crefname{defn}{Definition}{Definitions}
\crefname{notn}{Notation}{Notations}
\crefname{assmp}{Assumption}{Assumptions}
\crefname{thm}{Theorem}{Theorems}
\crefname{prp}{Proposition}{Propositions}
\crefname{lem}{Lemma}{Lemmas}
\crefname{cor}{Corollary}{Corollaries}
\crefname{conj}{Conjecture}{Conjectures}
\crefname{rmk}{Remark}{Remarks}
\crefname{exmp}{Example}{Examples}
\crefname{section}{Section}{Sections}
\crefname{subsection}{Subsection}{Subsections}
\crefname{para}{}{}
\crefname{appendix}{Appendix}{Appendices}
\crefname{subappendix}{Appendix}{Appendices}
\crefname{table}{Table}{Tables}

\newcommand{\bB}{\mathbb{B}}
\newcommand{\bC}{\mathbb{C}}

\newcommand{\bK}{\mathbb{K}}

\newcommand{\bM}{\mathbb{M}}

\newcommand{\bR}{\mathbb{R}}

\newcommand{\bT}{\mathbb{T}}

\newcommand{\bX}{\mathbb{X}}

\newcommand{\bZ}{\mathbb{Z}}
\newcommand{\cA}{\mathcal{A}}
\newcommand{\cB}{\mathcal{B}}

\newcommand{\cT}{\mathcal{T}}

\newcommand{\ba}{\mathbf{a}}
\newcommand{\bb}{\mathbf{b}}
\newcommand{\bc}{\mathbf{c}}

\newcommand{\bk}{\mathbf{k}}

\newcommand{\br}{\mathbf{r}}

\newcommand{\bv}{\mathbf{v}}
\newcommand{\bw}{\mathbf{w}}
\newcommand{\bx}{\mathbf{x}}
\newcommand{\by}{\mathbf{y}}
\newcommand{\bz}{\mathbf{z}}

\newcommand{\sD}{\mathscr{D}}

\newcommand{\sH}{\mathscr{H}}

\newcommand{\pt}{\mathrm{pt}}

\DeclareMathOperator{\K}{\mathrm{K}}
\DeclareMathOperator{\KR}{\mathrm{KR}}

\DeclareMathOperator{\Ad}{\mathrm{Ad}}

\newcommand{\tr}{\mathrm{tr}}

\author{Yosuke Kubota}
\address{Department of Mathematical Sciences, Shinshu University\\ 3-1-1 Asahi, Matsumoto, Nagano, 390-8621, Japan\\ and \\ RIKEN iTHEMS \\ 2-1 Hirosawa, Wako, Saitama, 351-0198, Japan}
\email{ykubota@shinshu-u.ac.jp}
\title[The bulk--dislocation correspondence]{The bulk--dislocation correspondence for weak topological insulators on screw--dislocated lattices}

\date{\today}
\begin{document}
\maketitle
\begin{abstract}
A weak topological insulator in dimension $3$ is known to have a topologically protected gapless mode along the screw dislocation. In this paper we formulate and prove this fact with the language of C*-algebra K-theory. The proof is based on the coarse index theory of the helical surface. 
\end{abstract}

\tableofcontents

\section{Introduction}\label{section:1}
The bulk-boundary correspondence is one of the fundamental issues in the theory of topological insulators. 
This principle states that a topological nature of a quantum system at the bulk ensures the existence of topologically protected boundary states of the same system to which a boundary is inserted. 
There are several theoretical frameworks demonstrating the bulk-boundary correspondence physically or mathematically. 
This paper deals with the one based on functional analysis and operator algebra, in which the bulk-boundary correspondence is understood as the boundary map of  C*-algebra K-theory.
For the researches in this direction, we refer the readers to e.g.~\cites{bellissardNoncommutativeGeometryQuantum1994,kellendonkEdgeCurrentChannels2002,prodanBulkBoundaryInvariants2016,bourneKtheoreticBulkedgeCorrespondence2017,kubotaControlledTopologicalPhases2017}.

The aim of this paper is to prove the \emph{bulk-dislocation correspondence} in this functional analytic setup. 
In short, this principle states that the topology of a $3$-dimensional Hamiltonian ensures the existence of topologically protected states localized at a screw dislocation. More specifically, a weak topological insulator in the $xy$-direction has a localized state along a screw dislocation at the $z$-axis. 
This phenomenon is already discovered in the literature of physics such as Ran--Zhang--Vishwanath~\cite{ranOnedimensionalTopologicallyProtected2009}, Teo--Kane~\cite{teoTopologicalDefectsGapless2010} and Imura--Takane--Tanaka~\cite{imuraWeakTopologicalInsulator2011}. There is also a mathematical approach to this problem by Hannabuss--Mathai--Thiang \cite{hannabussTdualitySimplifiesBulkboundary2016}*{Section 6}, in which the K-theory of the C*-algebra of Heisenberg groups are studied through the noncommutative T-duality. This approach shares a mathematical basis with ours but we deal with a different C*-algebra.

In a functional analytic formulation we are working in, a quantum system is characterized by its Hamiltonian operator. 
Under the 1-particle and tight-binding approximations, it is a bounded operator acting on the Hilbert space of lattices $\ell^2(\bZ^d; \bC^N)$, where $d$ is the space dimension and $N$ is the internal degree of freedom. 
We assume that $H$ is insulated, i.e., has a spectral gap at the Fermi energy (which is often assumed to be $0$). 
If we additionally assume that $H$ is of short-range and translation invariant, it is a self-adjoint invertible element in the C*-algebra $C^*_r(\bZ^d) \otimes \bM_N$, where $C^*_r(\bZ^d)$ is the group C*-algebra of the abelian group $\bZ^d$.
Such an operator is topologically classified by the $\K_0$-group of $C^*_r(\bZ^d)$. 
The corresponding boundary Hamiltonian lies in the C*-algebra $\cT \otimes C^*_r(\bZ^{d-1}) $, where $\cT$ denotes the Toeplitz algebra. 
The bulk-boundary correspondence is understood as the boundary map of the following Toeplitz exact sequence
\[ 0 \to \bK \otimes C^*_r(\bZ^{d-1}) \to \cT \otimes C^*_r(\bZ^{d-1}) \to C^*_r(\bZ^d) \to 0. \]
Indeed, the image of $[H] \in \K_0(C^*_r(\bZ^d))$ 
describes topologically protected boundary states of the corresponding boundary Hamiltonian.

Our formulation of the bulk-dislocation correspondence is given in a similar fashion. We define a C*-algebra exact sequence of the form
\[0 \to \bK \otimes C^*_r(\bZ)   \to \cA_{\bz} \to C^*_r(\bZ^3) \to 0, \]
where $\cA_{\bz}$ is the C*-algebra of observables on the dislocated lattice (a precise definition is given in \eqref{eq:ext}). Then the boundary map of C*-algebra K-theory sends a gapped phase without dislocation $[H] \in \K_0(C^*_r(\bZ^3))$ to the $\K_{-1}$-element describing the topologically protected dislocation state.
The main theorem of this paper, \cref{thm:main}, determines this homomorphism. Indeed, it turns out to send the $2$-dimensional Bott generator $\beta_{xy}$ in the $xy$-direction to the generator of $\K_1(\bK \otimes C^*_r(\bZ) )$ and other Bott generators to zero. 

The proof of our main theorem is based on the coarse geometry of the helical surface. 
As is discussed in \cref{section:discuss}, the bulk--dislocation correspondence is thought of as a bulk-boundary correspondence for Hamiltonians on the helical surface. 
This identification is given in a `large-scale' way, which matches with the philosophy of coarse geometry. 
In particular, the key technical lemmas for the proof, \cref{lem:sigma}, relies on a relatively new technique to lift a finite propagation operator on the covering space, which is originally given in \cite{gongGeometrizationStrongNovikov2008}. 
A remarkable point is that the proof is parallel to the author's study on the codimension 2 transfer map in higher index theory \cites{kubotaGromovLawsonCodimensionObstruction2020,kubotaCodimensionTransferHigher2021}, which is motivated from differential topology of manifolds.
Indeed, the simplest case of this codimension 2 transfer map, a map from the C*-algebra K-theory of the group C*-algebra of $\bZ^2$ to that of the trivial group, is the same thing as the proof of the bulk-dislocation correspondence given in this paper. 
A detail is discussed in \cref{rmk:codim2}.

This paper is organized as following. In \cref{section:formulation}, we give an operator-algebraic setup of the bulk-dislocation correspondence and state the main theorem. In \cref{section:proof}, we give a heuristic discussion introducing the strategy of the proof, summarize the foundation of coarse index theory, and give a proof of the main theorem. In \cref{section:remarks}, we list remarks on our main theorem, including several generalizations. 
In \cref{section:appendix}, we give a proof of the key lemma of this paper, originally given in \cite{kubotaCodimensionTransferHigher2021}, specified to our setting.  

\subsection*{Acknowledgement}
This work was supported by RIKEN iTHEMS and JSPS KAKENHI Grant Numbers 19K14544, JPMJCR19T2, 17H06461.

\section{A mathematical formulation}\label{section:formulation}
We start with formulating the problem with the language of functional analysis and C*-algebra. 
We define the C*-algebra of quantum mechanical observables on the lattice possessing the screw dislocation along the Burgers vector $\bz:=(0,0,1) \in \bZ^3$. 
The general case, i.e., the case that the Burgers vector is an arbitrary lattice vector, is discussed later in \cref{rmk:Burgers}.

Let $X:=\bR^2$, let $X_0:=X \setminus \{ (0,0)\}$ and let $\widetilde{X}_0$ be the universal covering of $X_0$. 
This is a $\bZ$-Galois covering, and hence is equipped with a free and proper $\bZ$-action. 
We choose $(1,0)$ as the basepoint and represent a point in $\widetilde{X}_0$ by the homotopy class $[\gamma]$ of smooth paths $\gamma(t) =(x(t),y(t)) \colon [0,1] \to X_0$ starting from $(1,0)$.
We embed $\widetilde{X}_0$ into $\bR^3$ in a $\bZ$-equivariant way, where the $\bZ$-action on $\bR^3$ is the shift by $1$ in the $z$-direction, as
\[\iota_\bz \colon \widetilde{X} \to \bR^3, \ \ \iota([\gamma ])= \Big(x(1), y(1) , \int_\gamma \gamma^* d\theta \Big),\]
where $d\theta = (xdy - ydx)/(x^2 + y^2)$ is a $1$-form on $X_0$. 
The image $\iota_\bz (\widetilde{X}_0)$, denoted by $\widetilde{X}_{\bz,0}$ hereafter, is the helical surface embedded in $\bR^3$. 
We define the dislocated lattice $\bX_{\bz , 0}$ as a lattice of $\widetilde{X}_{\bz,0} $, that is, 
\begin{align*}
    \bX_{\bz,0} :=& \widetilde{X}_0 \cap (\bZ \times \bZ \times \bR)  \\
=&\Big\{ \Big(x,y,z+ \frac{\arctan (y/x)}{2\pi}  \Big) \mid (x,y,z) \in (\bZ^2 \setminus \{ 0,0\}) \times \bZ  \Big\} _{\textstyle .}
\end{align*}
Set $\bX_\bz :=\bX_{\bz ,0} \cup (\{ 0\} \times \{ 0 \} \times \bZ) $.
This set is thought of as the configuration of atoms of a screw-dislocated $3$-dimensional material. The (1-particle tight-binding) Hamiltonian operator is a bounded operator in $\bB(\ell^2(\bX_\bz; \bC^N))$. 

For $\bv \in \bZ^2 \times \bR$, let $[\bv]$ stand for the vector in $\bX_\bz$ which is nearest to $\bv$ (if there are two such vectors, then $[\bv]$ stands for the one whose $z$-coordinate is larger). Note that $\bv$ and $[\bv]$ shares the $(x,y)$-coordinate. 
We use this notation to define three unitary operators $\widetilde{S}_x$, $\widetilde{S}_y$, $\widetilde{S}_z$ on $\ell^2(\bX_\bz)$ as
\begin{align}
    \widetilde{S}_x \delta _\bv := \delta_{[\bv+\bx]} , &&  \widetilde{S}_y \delta _\bv := \delta_{[\bv + \by]},&& \widetilde{S}_z\delta _\bv := \delta_{\bv + \bz},
\end{align}
where $\bx:=(1,0,0)$ and $\by:=(0,1,0)$. 
They are regarded as the dislocated lattice translation in the $x$, $y$, $z$-directions respectively. 
For a convenience of calculations, we identify these operators with an operator on the non-dislocated lattice $\bZ^3$ through the unitary $\Phi_\bz \colon \ell^2(\bZ^3) \to \ell^2(\bX_\bz)$ determined by $\Phi_\bz (\delta_\bv) := \delta_{[\bv]}$.  
\begin{lem}\label{lem:shift}
Let $\widetilde{S}_x$, $\widetilde{S}_y$, $\widetilde{S}_z$ and $\Phi_\bz$ be as above. We write $S_x$, $S_y$, $S_z$ for the lattice translation on $\ell^2(\bZ^3)$ in the $x$, $y$, $z$-directions. Then the following hold.
\begin{enumerate}
    \item We have
\begin{align*}
    \Phi_\bz^* \widetilde{S}_x \Phi_\bz &= S_x , \\
    \Phi_\bz^* \widetilde{S}_y \Phi_\bz &= S_y ( 1-P  + P S_z^* ) ,\\
    \Phi_\bz^* \widetilde{S}_z \Phi_\bz &= S_z ,
\end{align*}
where $P \in \bB(\ell^2(\bZ^3))$ denotes the projection onto $\ell^2(\bZ_{< 0} \times \{ 0 \} \times \bZ) )$.
\item We have $[\widetilde{S}_x,\widetilde{S}_z] = 0$, $[\widetilde{S}_y,\widetilde{S}_z] =0$ and
\[[\widetilde{S}_x, \widetilde{S}_y] = \Phi _\bz (S_xS_y p \otimes 1)  (1-S_z^*) \Phi_\bz ^* , \]
where $p \in \bK(\ell^2(\bZ^2))$ denotes the (rank one) projection onto $\ell^2(\{ (-1,0) \})$.  
\end{enumerate}
\end{lem}
\begin{proof}
The claim (1) is obvious from the definition, and (2) follows from (1).
\end{proof}

\begin{defn}
We write $\cA_{\bz}$ for the C*-subalgebra of $\bB(\ell^2(\bX_\bz))$ generated by $\widetilde{S}_x$, $\widetilde{S}_y$, $\widetilde{S}_z$ and the C*-subalgebra $\Phi_\bz(\bK(\ell^2(\bZ^2)) \otimes C^*_r(\bZ)) \Phi_\bz^*$. 
\end{defn}

There is an abstract characterization of the C*-algebra $\cA_\bz$. To this aim we introduce two standard terminologies in coarse geometry \cite{roeLecturesCoarseGeometry2003}, both of which will be used in \cref{section:proof} again. 
\begin{enumerate}\label{page:finprop}
    \item The support of an operator $T \in \bB(\ell^2(\bX_\bz) )$ is defined as the support of the kernel function of $T$. In other words, 
    \[\mathrm{supp}(T):=\{ (\bv, \bw) \in \bX_\bz^2 \mid p_\bv T p_{\bw} \neq 0 \}, \]
    where $p_\bv$ denotes the rank one projection onto $\ell^2(\{ \bv\})$. 
    \item An operator $T \in \bB(\ell^2(\bX_\bz))$ is of \emph{finite propagation} if there is $R>0$ such that $\mathrm{supp} (T) $ is included to the $R$-neighborhood of the diagonal, in other words, $p_\bv T p_\bw =0$ for any $\bv,\bw \in \bX_\bz$ such that $d(\bv, \bw) >R$ (the infimum of such $R>0$ is called the propagation of $T$). 
\end{enumerate}
We call an operator $T \in \bB(\ell^2(\bX_\bz))$ \emph{$xy$-translation invariant away from the $z$-axis} if there is $R>0$ such that the commutators $[T,\widetilde{S}_x]$, $[T,\widetilde{S}_y]$ are both supported in $U_R \times U_R$, where $U_R$ denotes the $R$-neighborhood of the $z$-axis in $\bX_\bz$. 
\begin{lem}\label{lem:OA}
The closure of the set of operators satisfying the conditions
\begin{enumerate}
    \item of finite propagation, 
    \item $xy$-translation invariant away from the $z$-axis, and
    \item translation invariant in the $z$-direction, i.e., $[T,\widetilde{S}_z]=0$,
\end{enumerate}
coincides with $\cA_\bz$.
\end{lem}
\begin{proof}
A finite propagation operator $T$ is decomposed into a finite sum
\[ T= \sum_{l,m,n} f_{l,m,n}(\bv) \widetilde{S}_x^l \widetilde{S}_y^m \widetilde{S}_z^n,   \]
where each $f_{l,m,n}(\bv)$ is a bounded function on $\bX_\bz$. 
Under this expression of $T$, the condition (3) is equivalent to $f_{l,m,n} (\bv) =f_{l,m,n}(\bv + \bz)$ for all $\bv \in \bX_\bz$. Similarly, the condition (2) corresponds to $f_{l,m,n}(\bv) = f_{l,m,n}([\bv + \bx]) = f_{l,m,n}([\bv + \by])$ for any $\bv \in \bX_\bz \setminus U_R$. This shows that there is $f_{l,m,n}^\infty \in \bC$ such that each $f_{l,m,n}(\bv) - f_{l,m,n}^\infty $ is supported in $U_R$, and hence $T - \sum _{l,m,n} f_{l,m,n}^\infty \widetilde{S}_x^l \widetilde{S}_y^m \widetilde{S}_z^n $ is in the subalgebra $\Phi_\bz (\bK \otimes C^*_r(\bZ)) \Phi_\bz^*$. This finishes the proof.   
\end{proof}

It is easily verified from \cref{lem:shift} (1) that the C*-subalgebra $\Phi_\bz(\bK(\ell^2(\bZ^2))\otimes C^*_r(\bZ)) \Phi_\bz^*$ is an ideal of $\cA_\bz$. 
For simplicity of notations, we write this ideal shortly as $\bK \otimes C^*_r(\bZ)$. 
There is a $\ast$-homomorphism $\phi \colon \cA_\bz \to C^*_r(\bZ^3)$ determined by $\phi (\widetilde{S}_\bx) = S_\bx$, $\phi (\widetilde{S}_\by) = S_\by$ and $\phi (\widetilde{S}_\bz) = S_\bz$, whose kernel is $\bK \otimes C^*_r(\bZ)$. That is, there is an extension of C*-algebras
\begin{align} 
0 \to \bK \otimes C^*_r(\bZ)   \to \cA_{\bz} \to C^*_r(\bZ^3) \to 0. \label{eq:ext}
\end{align}
This induces the boundary map
\[\partial \colon \K_0(C^*_r(\bZ^3)) \to \K_1(\bK \otimes C^*_r(\bZ)) \cong \K_1(C^*_r(\bZ))\]
in C*-algebra K-theory. 

The K-theory of the group C*-algebra $C^*_r(\bZ^3)$ is isomorphic to the topological K-theory of the $3$-dimensional torus, and hence 
\begin{align}
    \K_n(C^*_r(\bZ^3)) \cong \K_n \oplus (\K_{n-1})^{\oplus 3} \oplus (\K_{n-2})^{\oplus 3} \oplus \K_{n-3}, \label{eq:Kgrp}
\end{align}
where $\K_* :=\K_*(\bC)$ (which is isomorphic to the topological K-group $\K^{-*}(\pt)$). 
Since $\K_0 \cong \bZ$ and $\K_1 \cong 0$, and hence $\K_*(C^*_r\bZ^3)$ is generated by $8$ elements $\beta_\emptyset , \beta_x, \beta_y, \beta_z , \beta_{xy}, \beta_{yz}, \beta_{zx}, \beta_{xyz}$. That is, 
\begin{align*}
\K_0(C^*_r\bZ^3) &\cong \bZ \beta_\emptyset  \oplus \bZ \beta_{xy} \oplus \bZ \beta_{yz} \oplus \bZ\beta _{zx},\\
\K_1(C^*_r\bZ^3) &\cong  \bZ \beta_{x} \oplus \bZ \beta_{y} \oplus \bZ\beta _{z} \oplus \bZ \beta_{xyz}.
\end{align*}

\begin{rmk}\label{rmk:real}
In the theory of topological insulators, the Real K-theory plays an important role as well as the complex K-theory. 
The Real version of C*-algebra K-theory is defined for C*-algebras equipped with a Real structure, i.e., an antilinear $\ast$-automorphic involution $\bar{\cdot} \colon A \to A$.
We impose the Real structure onto $\cA_\bz$ induced from the complex conjugation on $\ell^2(\bX_\bz)$. Note that this complex conjugation satisfy 
\[ \overline{\widetilde{S}_\bx}=\widetilde{S}_\bx,  \ \ \ \overline{\widetilde{S}_\by}=\widetilde{S}_\by, \ \ \ \overline{\widetilde{S}_\bz}=\widetilde{S}_\bz.\] 
Hence the quotient $\ast$-homomorphism $\cA_\bz \to C^*_r(\bZ^3)$ preserves the Real structure if we impose the Real structure on $C^*_r(\bZ^3)$ determined by $\overline{S_\bx} = S_\bx$, $\overline{S_\by} = S_\by$, $\overline{S_\bz} = S_\bz$. 
Through the Gelfand--Naimark duality, this Real C*-algebra is isomorphic to the continuous function algebra $C(\bT^3)$ onto the $3$-dimensional Brillouin torus, with the Real structure $\bar{f}(\bk) = \overline{f(-\bk)}$ for any $\bk \in \bT^3$. 

For a Real C*-algebra $A$, the Real K-theory $\KR_0(A)$ is defined as the group completion of the monoid of homotopy classes of conjugation-invariant projections in $\bigcup_n \bM_n(A)$. When $A = C^*_r\bZ^3 \cong C(\bT^3)$ as above, the Real K-group $\KR_*(C^*_r(\bZ^3))$ is isomorphic to the KR-group $\KR^{-*}(\bT^3, \tau)$, where $\tau (\bk) = -\bk$. Hence it is isomorphic to the direct sum of $8$ groups
\begin{align*}
    \KR_n(C^*_r\bZ^3) \cong \KR_n \oplus (\KR_{n-1})^{\oplus 3} \oplus (\KR_{n-2})^{\oplus 3} \oplus \KR_{n-3} \label{eq:Kgrp}
\end{align*}
in the same way as \eqref{eq:Kgrp}. 
Each direct summand is generated by a single element $\beta_{\emptyset }^n$, $\beta_x^{n-1}$, $\beta_y^{n-1}$, $\beta_z^{n-1}$, $\beta_{xy}^{n-2}$, $\beta_{yz}^{n-2}$, $\beta_{zx}^{n-2}$ and $\beta_{xyz}^{n-3}$ respectively. 
Here the element $\beta _{\bullet}^{l} \in \KR_l$ is a free generator if $l \equiv 0, 4$ mod $8$, a $2$-torsion element if $l \equiv 1,2$ mod $8$, and otherwise $0$. 
\end{rmk}

The goal of this paper is to determine the image of these generators by the boundary map $\partial$ associated to the exact sequence \eqref{eq:ext}.

\begin{thm}\label{thm:main}
The boundary map $\partial  \colon \K_*(C^*_r(\bZ^3) ) \to \K_{*-1}(\bK \otimes C^*_r(\bZ))$ associated to the extension \eqref{eq:ext}
sends the Bott element $\beta_{xy}$ in the $xy$-direction to the generator $\beta$ of $\K_1(C^*_r(\bZ))$ and the other Bott generators to zero. 
\end{thm} 

The essential part of this theorem is to determine $\partial (\beta_{xy})$. 
Before that, we prove the remainder. 
\begin{lem}
The boundary map $\partial $ sends $\beta_\emptyset $, $\beta_x$, $\beta_y$, $\beta_z$, $\beta_{yz}$ and $\beta _{xz}$ to zero. 
\end{lem}
\begin{proof}
Let $i_{yz} \colon C^*_r(\bZ^2) \to C^*_r(\bZ^3)$ denote the inclusion to the $yz$-component. As is seen in \cref{lem:shift} (2), the unitaries $\widetilde{S}_y$ and $\widetilde{S}_z$ commute. 
Hence the map $\tilde{i}_{yz}(S_y)=\widetilde{S}_y $ and $\tilde{i}_{yz}(S_z) =\widetilde{S}_z$ gives rise to a $\ast$-homomorphism
\[\tilde{i}_{yz} \colon C^*_r(\bZ^2) \to \cA_\bz \]
such that the diagram
\[
\xymatrix{
&&  & C^*_r(\bZ^2) \ar[d]^{i_{yz}} \ar@{.>}[dl]_{\tilde{i}_{yz}}   \\ 
0 \ar[r] & \bK \otimes C^*_r(\bZ) \ar[r] & \cA_\bz \ar[r]^\phi  & C^*_r(\bZ^3) \ar[r] & 0 
}
\]
commutes. Therefore we get 
\[ \partial (\beta_{yz}) = \partial (i_{yz}(\beta)) = (\partial \circ \phi_*) ( \tilde{i}_{yz}(\beta))=0 ,\]
where $\partial \circ \phi_* = 0$ comes from the C*-algebra K-theory long exact sequence. 
This discussion also shows that $\partial \beta_\emptyset  =0$, $\partial \beta_y=0$ and $\partial \beta_z =0$, because $\beta_\emptyset , \beta_y, \beta_z$ are all in the image of $(\tilde{i}_{yz})_*$. 

The claims $\partial (\beta_{xz}) =0$ and $\partial (\beta_x) =0$ are also proved in the same way. 
\end{proof}

The claim $\partial \beta_{xyz} =0$ is reduced to $\partial \beta_{xy}=\beta $ in the following way. 
Let $\iota_{xy} \colon C^*_r(\bZ^2) \to C^*_r(\bZ^3 )$ denote the inclusion to the $xy$-component. 
Let $\cB_{\bz}$ denote the pull-back of $C^*_r(\bZ^2) \subset C^*_r(\bZ^3)$. Then there is an exact sequence 
\begin{align}
    0 \to \bK \otimes C^*_r(\bZ)  \to \cB_{\bz } \xrightarrow{\phi} C^*_r(\bZ^2) \to 0. \label{eq:ext2d}
\end{align}
\begin{lem}
Suppose that $\partial \beta_{xy} = \beta \in \K_{-1}(\bK \otimes C^*_r (\bZ) )$ is verified. Then $\partial \beta_{xyz} =0$ holds. 
\end{lem}
\begin{proof}
The map $T \otimes u^k \mapsto T \widetilde{S}_z^k$ induces a $\ast$-homomorphism $C^*_r(\bZ) \otimes \cB_\bz \to \cA_\bz$. This extends to a commutative diagram
\[\mathclap{
\xymatrix{
0 \ar[r] & C^*_r(\bZ) \otimes (\bK \otimes C^*_r(\bZ)) \ar[r] \ar[d]^m & C^*_r(\bZ) \otimes \cB_\bz \ar[r] \ar[d] & C^*_r(\bZ) \otimes C^*_r(\bZ^2) \ar[r] \ar[d] & 0 \\
0 \ar[r] &\bK \otimes C^*_r(\bZ)  \ar[r]& A_\bz \ar[r] & C^*_r(\bZ^3) \ar[r] & 0,
}}
\]
where the right vertical map is an isomorphism and the left vertical map $m$ is identified, through the Gelfand-Naimark duality $C^*_r(\bZ) \cong C(\bT)$, with the pull-back with respect to the diagonal embedding $\bT \to \bT \times \bT$. 

By the K\"{u}nneth theorem \cite{rosenbergKunnethTheoremUniversal1986}, there is an isomorphism of the K-group of $\K_*(A \otimes C^*_r(\bZ)) \cong \K_*(A) \otimes \K_*(C^*_r(\bZ))$ which is functorial and is compatible with the K-theory long exact sequence. Hence we get
\[\partial (\beta \otimes \beta _{xy} ) = \beta \otimes \partial \beta_{xy}  = \beta \otimes \beta. \]
This completes the proof as
\[\partial (\beta_{xyz}) = m_*\circ \partial (\beta_{xy} \otimes \beta) = m_*(\beta \otimes \beta) =0. \qedhere \]
\end{proof}

\section{Proof of the main theorem}\label{section:proof}
In this section we give a proof of the essential part of our main theorem. Here we focus on complex K-theory, and show that the boundary map $\partial$ of the K-theory long exact sequence of \eqref{eq:ext} sends the Bott generator $\beta_{xy}$ of the $xy$-component of $\K_0(C^*_r(\bZ^3))$ to the generator $\beta $ of $\K_{-1}(\bK \otimes C^*_r(\bZ) )$. 
We remark that the proof given in this section also works for the corresponding result in Real K-theory. This point is discussed in \cref{rmk:TI}.

\subsection{Discussion}\label{section:discuss}
Before going to the formal proof, we shortly illustrate how the weak topology of a bulk Hamiltonian in the $xy$-direction induces a dislocation-localized state along the $z$-axis. 
Let $H$ be a Hamiltonian on the $2$-dimensional standard lattice, which is of the form 
\[ H := \sum_{(n,m) \in \bZ^2} A_{nm} S_x^nS_y^m \in \bM_N \otimes C^*_r(\bZ^2).  \]
Here, each $A_{nm}$ is an $N \times N$ matrix satisfying $A_{nm}^*=A_{mn}$. 
It is usually assumed to be of short-range, i.e., the coefficient decays exponentially as $\| A_{nm}\| \leq C_1e^{-C_2 \sqrt{n^2+m^2}}$ for some $C_1, C_2>0$. 
For simplicity of the discussion, we impose a stronger assumption; the right hand side is a finite sum. This corresponds to the finite propagation condition given in page \pageref{page:finprop}. 
Moreover, we assume that $H$ has a spectral gap at the Fermi level $\mu =0 $. 
This $H$, regarded as a $3$-dimensional quantum observable with the same presentation through the inclusion of C*-algebras $C^*_r(\bZ^2) \subset C^*_r(\bZ^3)$, gives a model of an $xy$-weak topological insulator of type A. 
As an actual operator acting on the Hilbert space $\ell^2(\bZ^3;\bC^N)$, this $H$ is the superposition of infinitely many layers of $2$-dimensional Hamiltonians in the $z$-direction. 

The corresponding Hamiltonian on the screw-dislocated lattice $\bX_\bz$ is 
\[ H_\bz := \sum A_{nm} \widetilde{S}_x^n \widetilde{S}_y^m. \]
This is a lift of $H \in C^*(\bZ^2) \otimes \bM_N$ to $A_\bz \otimes \bM_N$ with respect to the quotient $\phi$ in \eqref{eq:ext2d}. 
Therefore, by definition of the boundary map in $\K$-theory, the image of the boundary map $\partial [H] \in \K_*(\bK \otimes C^*_r(\bZ))$ is represented by a unitary
\[ U_\bz := -\exp [- \pi i \chi (H_\bz)] \in (\bK \otimes C^*_r(\bZ))^+, \]
where $\chi \colon \bR \to [-1,1]$ is a continuous function such that $\chi|_{(-\infty, -\varepsilon ]} \equiv -1$ and $\chi|_{[\varepsilon , \infty)} \equiv 1$.
This unitary extracts the spectrum of the dislocated Hamiltonian $H_\bz$ inside the bulk spectral gap. 

We consider replacing this $H_\bz$ with another lift of $H$ in $A_\bz$. 
Let $\widetilde{X}_{\bz,R}$ and $\bX_{\bz,R}$ denote the complement of the $R$-neighborhood of the $z$-axis in $\widetilde{X}_{\bz,0}$ and $\bX_\bz$ respectively. Let $\Pi_R$ denote the projection onto $\ell^2(\bX_{\bz,R};\bC^N) \subset \ell^2(\bX_\bz;\bC^N)$. The operator
\[ H_{\bz , R} := \Pi _R H_\bz \Pi_R + (1-\Pi_R) \in A_\bz \]
is also a lift of $H$ since $H_{\bz ,R} - H_\bz \in \bK \otimes C^*_r(\bZ)$. 
If $R>0$ is sufficiently larger than the range of the Hamiltonian $H$, then the above $H_{\bz,R}$ is regarded as an operator on a lattice $\bX_{\bz, R}$ of the helical surface $\widetilde{X}_{\bz,R}$.
The helical surface $\widetilde{X}_{\bz,R}$ is a noncompact $2$-dimensional manifold with boundary. The boundary $\partial \widetilde{X}_{\bz,R}$ is a helix, and hence is $\bZ$-equivariantly and quasi-isometrically homeomorphic to the real line $\bR$ (on which $\bZ$ acts as the shift by $1$).
Then the bulk-boundary correspondence for $\widetilde{X}_{\bz, R}$ will show that a non-trivial topology of $H$ implies an edge current along the boundary helix. 
\input{tikz.tex}
In order to shape the discussion into a rigorous proof, there are two difficulties. One is that the helical surface no longer have the translation symmetry in the $xy$-direction. Another is that the radius $R>0$ of the boundary helix is chosen after the range of $H$ is fixed, and hence there is no simultaneous construction which works for all finite range Hamiltonians. 
Both of these two problems are resolved by working in the framework of coarse index theory.

\subsection{Coarse index theory}
Here we enumerate a minimal list of the foundation of coarse geometry and coarse index theory used in the paper. For a more detail, we refer the reader to \cites{roeIndexTheoryCoarse1996,roeLecturesCoarseGeometry2003,higsonAnalyticHomology2000,willett_yu_2020}.  

Let $W$ be a locally compact metric space equipped with a free proper action of a discrete group $\Gamma$ (the translation action of $\bZ^2$ in the $xy$-direction on $X=\bR^2$ and the translation action of $\bZ$ in the $z$-direction on $\widetilde{X}_{\bz,R}$ are the examples of our interest). 
Let $\pi \colon C_0(W) \to \bB(\sH)$ be a $\Gamma$-equivariant $\ast$-homomorphism which is ample, i.e., no non-zero function $f \in C_0(W)$ acts as a compact operator. 
When $W$ is a Riemannian manifold, we choose as $(\pi, \sH)$ the multiplication representation of $C_0(W)$ onto the $L^2$-space $L^2(W )$ with respect to the $\Gamma$-invariant volume form on $W$.  
When $W$ is a discrete metric space, we choose as $\sH$ the infinite direct sum $\ell^2(W)^{\oplus \infty}$ on which $C_0(W)$ acts by multiplication. 
We define the notion of support and propagation for an operator $T \in \bB(L^2(W))$ in the same way as page \pageref{page:finprop} (for the precise definition, see e.g., \cite{higsonAnalyticHomology2000}*{Definition 6.3.3}). 

The invariant Roe algebra $C^*(W)^\Gamma$ is the closure of the $\ast$-algebra $\bC[W]^\Gamma$ of $\Gamma$-invariant operators on $\sH$ which are of finite propagation and locally compact, i.e., $Tf, fT \in \bK(\sH)$ for any $f \in C_c(W)$. 
When $\Gamma$ is trivial, this C*-algebra is called the Roe algebra and written as $C^*(W)$. 
Moreover, for a $\Gamma$-invariant subspace $V \subset W$, the ideal $C^*(V \subset W)$ is defined as the closure of $\Gamma$-invariant, finite propagation, locally compact operators whose support is included to an $R$-neighborhood of $V \times V$ for some $R>0$.

Let $D^*(W)^\Gamma$ denote the closure of the $\ast$-algebra $D^*_{\mathrm{alg}}(W)^\Gamma$ of $\Gamma$-invariant bounded operators on $L^2(W)$ which are of finite propagation and quasi-local, i.e., the commutator $[T,f]$ is a compact operator for any $f \in C_c(W)$. 
We also define the ideal $D^*(V \subset W)^{\Gamma}$ of $D^*(W)^\Gamma$ as the closure of the set of finite propagation quasi-local operators $T$ such that $T f , fT \in \bK(\sH )$ for any $f \in C_0(W)$ such that $f|_V \equiv 0$. 

The inclusions $C^*(W)^\Gamma  \subset D^*(W)^\Gamma$ and $C^*(V \subset W) \subset D^*(V \subset W)$ are ideals. We write $Q^*(W)^\Gamma$ and $Q^*(V \subset W)^\Gamma $ for the quotient $ D^*(W)^\Gamma / C^*(W)^\Gamma$ and $D^*(V \subset W)^{\Gamma} /C^*(V\subset W)^\Gamma $ respectively.

\begin{rmk}\label{rmk:coarse}
Here we list some basic facts on the K-theory of these coarse C*-algebras which will be used in the proof of our main theorem. 
\begin{enumerate}
    \item Let $\Gamma$ acts on $W$ cocompactly. By choosing a Borel subset $U \subset W$ such that $W = \bigsqcup_{g \in \Gamma} g \cdot U$, the Hilbert space $L^2(W)$ is identified with $\ell^2(\Gamma) \otimes L^2(U) $. 
    This induces a $\ast$-isomorphism of C*-algebras $C^*(W)^{\Gamma} \cong C^*_r(\Gamma ) \otimes \bK(L^2(U))$ (\cite{roeIndexTheoryCoarse1996}*{Lemma 5.14}). 
    \item For any $\Gamma$-invariant subspace $V \subset W$, the Roe algebras $C^*(V \subset W)^\Gamma$ and $ C^*(V)^\Gamma $ have the same K-theory. Also, $Q^*(V \subset W)^\Gamma$ and $ Q^*(V)^\Gamma $ have the same K-theory (\cite{siegelHomologicalCalculationsAnalytic2012}*{Proposition 4.3.34}). 
    \item When $W$ is an even dimensional complete Riemannian manifold with a spin structure, the Dirac operator $D $ determines a K-theory class $[W] \in \K_1(Q^*(W)^\Gamma)$ called the Dirac fundamental class. In the same way, if $W$ is a free proper $\Gamma$-manifold with $\Gamma$-invariant boundary, the relative fundamental class $[W , \partial W] \in \K_1(Q^*(W)^\Gamma /Q^*(\partial W \subset W)^\Gamma )$ is defined. 
    \item 
    The boundary map 
    \[\partial \colon \K_*(Q^*(W)^\Gamma /Q^*(\partial W \subset W)^\Gamma ) \to \K_{*-1}(Q^*(\partial W \subset W)^\Gamma ) \]
    sends $[W , \partial W]$ to the fundamental class $[\partial W]$ of the boundary (this fact is known as the `boundary of Dirac is Dirac' principle, see e.g., \cite{higsonAnalyticHomology2000}*{Proposition 11.2.15}). 
    \item The boundary map 
    \[\partial \colon \K_1(Q^*(W)^\Gamma ) \to \K_0(C^*(W)^\Gamma )\]
    is called the equivariant coarse index map and denoted by $\mathop{\mathrm{Ind}}$. When $W =\bR^n$ and $\Gamma = \bZ^n$ acting on $W$ by the translation, then the equivariant coarse index is the same as the family index (\cite{baumClassifyingSpaceProper1994}*{Example 3.11}), and in particular sends the fundamental class $[\bR^n] \in \K_1(Q^*(\bR^n)^{\bZ^n} )$ to the Bott generator of the top degree in $\K_0(C^*(\bR^n)^{\bZ^n}) \cong \K_0(C^*_r(\bZ^n)) \cong \K^0(\bT^n)$. 
\end{enumerate}
\end{rmk}

\subsection{Equivariant coarse geometry of helical surfaces}
Here we apply the facts listed in the previous subsection to the $\bZ$-equivariant coarse geometry of the helical surface. 
The following lemma, which lift a finite propagation operator on $X = \bR^2$ to the covering space $\widetilde{X}_{\bz, R}$, is a key ingredient of the proof of \cref{thm:main}. 
\begin{lem}\label{lem:sigma}
There is a $\ast$-homomorphism
\[ s \colon C^*(X) \to C^*(\widetilde{X}_{\bz,1})^\bZ/C^*(\partial \widetilde{X}_{\bz,1} \subset \widetilde{X}_{\bz,1})^\bZ, \]
 which extends to 
\[ s \colon D^*(X) \to D^*(\widetilde{X}_{\bz,1})^\bZ/D^*(\partial \widetilde{X}_{\bz,1} \subset \widetilde{X}_{\bz,1})^\bZ. \]
They induce a corresponding $\ast$-homomorphism between $Q^*$ coarse C*-algebras.
Moreover, the image $s_*[X]$ of the fundamental class of $X$ is $[\widetilde{X}_{\bz, 1} , \partial \widetilde{X}_{\bz, 1}]$. 
\end{lem}
A proof of this lemma is given in a more general geometric setting in \cite{kubotaCodimensionTransferHigher2021}. 
In \cref{section:appendix}, we give a full proof which is specific to our setting. 
Here we only sketch the construction of $s$. 

\begin{para}\label{para:lift}
Let $K $ be a bounded operator on $L^2(X)$ with propagation less than $R>0$. Moreover, we assume that $K$ is represented by convolution with a kernel function $k \colon X \times X \to \bC$ as $K\xi (x) = \int _X k(x,y)\xi(y)dy $. The set of such operators is dense in $C^*(X)$. 
We define the function $\widetilde{k}_R \colon \widetilde{X}_{\bz,R} \times \widetilde{X}_{\bz ,R} \to \bC$ as
\begin{align}
    \widetilde{k}_R(\tilde{x}, \tilde{y}) = \begin{cases} k(\pi (\tilde{x}), \pi(\tilde{y})) & \text{ if $d(\tilde{x}, \tilde{y}) < R$, }\\ 0 & \text{ otherwise,}\end{cases} \label{eq:kernel}
\end{align}
where $d$ denotes the metric on $\widetilde{X}_{\bz,R}$ induced from its Riemannian metric, and set 
\begin{align}
\widetilde{K}_R \xi(\tilde{x}):= \int_{\tilde{y} \in \widetilde{X}_{\bz,R}}\widetilde{k}_R(\tilde{x}, \tilde{y})\xi(\tilde{y}) d\tilde{y}. \label{eq:lift}
\end{align}
This $\widetilde{K}_R$ determines a bounded operator (this is a non-trivial part of the proof), which is locally compact and has propagation less than $R$. For $0 < R < S$, the difference $\widetilde{K}_S - \widetilde{K}_R$ is supported in the $(S+R)$-neighborhood of $\partial \widetilde{X}_{\bz,1}$, and hence lies in the ideal $C^*(\partial \widetilde{X}_{\bz, 1} \subset \widetilde{X}_{\bz, 1})$. 
Therefore, $s(K) := \widetilde{K}_R$ gives rise to a well-defined linear map from a dense subalgebra of $C^*(X)$ to $C^*(\widetilde{X}_{\bz, 1})^\bZ / C^*(\partial \widetilde{X}_{\bz, 1} \subset \widetilde{X}_{\bz,1})^\bZ $. 
Note that it remains to prove that this $s$ extends to $C^*(\widetilde{X}_{\bz, 1})^\bZ $, in other words, $s$ is a bounded linear map. 
\end{para}

Next we relate this $s$ with the extension \eqref{eq:ext2d}. 
Let $\psi_{\mathbf{0}} $ be an $L^2$-function on $L^2(X)$ supported in the $\varepsilon$-neighborhood of $\mathbf{0}$. For $\bv \in \bZ^2$, let $\psi_{\bv}(x):=\psi_{\mathbf{0}}(x-\bv)$. 
Then the Hilbert subspace $\bigoplus \bC \cdot \psi_{\bv} \subset L^2(X)$ is identified, as unitary representations of $\bZ^2$, with $\ell^2(\bZ^2)$. We write the unitary identifying these Hilbert spaces as $V$. 
Moreover, for each $\tilde{\bv} \in \bX_{\bz}$, let $\psi_{\tilde{\bv}}$ denote the restriction of the pull-back $\pi^* \psi _{\bv}$ to the $\varepsilon $-neighborhood of $\tilde{\bv}$.   
Then the Hilbert subspace $\bigoplus \bC \cdot \psi_{\tilde{\bv}} \subset L^2(\widetilde{X}_{\bz,1})$ is identified with $\ell^2(\bX_\bz)$ by a unitary, denoted by $\widetilde{V}$. 

Now, by definition of $s$ given in \cref{lem:sigma}, we have
\[s(V S_{\bx} V^*) = s\Big( \sum _{\bv} |\psi_{\bv + \bx} \rangle \langle \psi_\bv | \Big) =  \sum_{\tilde{\bv}} | \psi_{[\tilde{\bv} + \bx]} \rangle \langle \psi_{\tilde{\bv}} | = \widetilde{V} \widetilde{S}_{\bx} \widetilde{V}^* \]
modulo $C^*(\partial \widetilde{X}_{\bz,1} \subset \widetilde{X}_{\bz,1})$. Similarly, we have $s(V S_{\by} V^*) = \widetilde{V} \widetilde{S}_\by \widetilde{V}^*$ modulo $C^*(\partial \widetilde{X}_{\bz,1} \subset \widetilde{X}_{\bz,1})$. 
They mean that $\Ad (\widetilde{V})$ sends $\cB_\bz$ to $C^*(\widetilde{X}_{\bz,1})^\bZ$ and  the diagram
\[ 
\xymatrix{
0 \ar[r] & \bK \otimes C^*_r(\bZ) \ar[r] \ar[d]^{\Ad (\widetilde{V})} & \cB_\bz \ar[r] \ar[d]^{\Ad (\widetilde{V})} & C^*_r(\bZ^2) \ar[r] \ar[d]^{s \circ F \circ \Ad (V)} & 0 \\
0 \ar[r] & C^*(\partial \widetilde{X}_{\bz, 1} \subset \widetilde{X}_{\bz, 1})^\bZ \ar[r]  & C^*(\widetilde{X}_{\bz, 1})^\bZ  \ar[r]  & \frac{C^*(\widetilde{X}_{\bz, 1})^\bZ}{ C^*(\partial \widetilde{X}_{\bz, 1} \subset \widetilde{X}_{\bz, 1})^\bZ } \ar[r]  & 0 
}
\]
commutes, where $F \colon C^*(X)^{\bZ^2} \to C^*(X)$ denotes the inclusion. This induces the commutative diagram of K-theory
\begin{align}
\begin{split}
\xymatrix@C=4em{
\K_0(C^*_r(\bZ^2)) \ar[r]^{\partial } \ar[d]^{\Ad (V)} & \K_{-1}(\bK \otimes C^*_r(\bZ)) \ar[d]^{\Ad (\widetilde{V})}  \\
\K_0(C^*(X)^{\bZ^2} ) \ar[r]^{\partial \circ s_*  \circ F_* \ \ \ \ \ \ \ \ } & \K_{-1}(C^*(\partial \widetilde{X}_{\bz, 1} \subset \widetilde{X}_{\bz, 1})^\bZ). \\
}
\end{split}
\label{eq:diagram}
\end{align}
Moreover, by \cref{rmk:coarse} (1), the vertical maps in \eqref{eq:diagram} are isomorphisms.

\begin{proof}[Proof of \cref{thm:main}]
By \cref{lem:sigma}, we have the commutative diagram
\[\mathclap{
\xymatrix{
\K_1(Q^*(X)^{\bZ^2}) \ar[r]^{F_*} \ar[d]^{\mathop{\mathrm{Ind}}}  &\K_1(Q^*(X)) \ar[r]^{s_* \hspace{4ex}} \ar[d]^{\mathop{\mathrm{Ind}}} & \K_1\Big(\frac{Q^*(\widetilde{X}_{\bz, 1})^\bZ}{Q^*(\partial \widetilde{X}_{\bz, 1} \subset \widetilde{X})^\bZ_{\bz, 1}} \Big) \ar[r]^\partial   \ar[d]^{\mathop{\mathrm{Ind}}} & \K_0(Q^*(\partial \widetilde{X}_{\bz, 1})^\bZ ) \ar[d]^{\mathop{\mathrm{Ind}}} \\
\K_0(C^*(X)^{\bZ^2}) \ar[r]^{F_*} & \K_0(C^*(X)) \ar[r]^{s_* \hspace{4ex}} & \K_0\Big( \frac{C^*(\widetilde{X}_{\bz, 1})^\bZ}{C^*(\partial \widetilde{X}_{\bz, 1} \subset \widetilde{X}_{\bz, 1})^\bZ} \Big) \ar[r]^\partial  & \K_{-1}(C^*(\partial \widetilde{X}_{\bz, 1})^\bZ ) .
}}
\]
By \eqref{eq:diagram}, it suffices to show the composition $\partial \circ s_* \circ F_*$ in the second row maps the Bott generator $\beta_{xy} \in \K_0(C^*(X)^{\bZ^2})$ to $\beta \in \K_{-1}(C^*(\partial \widetilde{X}_{\bz, 1})^\bZ )$.
This is checked as
\begin{align*}
    (\partial \circ s_* \circ F_*) (\beta_{xy}) &= (  \partial \circ  s_* \otimes F_* \circ  \mathop{\mathrm{Ind}})([X])\\
    &= ( \mathop{\mathrm{Ind}}\circ  \partial \circ  s_* \otimes F_* )([X]) \\
    &=(\mathop{\mathrm{Ind}} \circ \partial) [\widetilde{X}_{\bz,1}, \partial \widetilde{X}_{\bz,1}] \\
    &= \mathop{\mathrm{Ind}} [\partial \widetilde{X}_{\bz,1}] = \beta.
\end{align*}
Here the first and the last equalities are due to \cref{rmk:coarse} (5), the third equality follows from \cref{lem:sigma} and the forth equality is the `boundary of Dirac is Dirac' principle exposed in \cref{rmk:coarse} (4). 
\end{proof}

\section{Miscellaneous remarks}\label{section:remarks}
We finish the paper by a list of remarks. 
\subsection{General Burgers vector}\label{rmk:Burgers}
\cref{thm:main} is generalized to a dislocated lattice with a general Burgers vector. Let us consider the Burgers vector $\bb:=(b_x,b_y,b_z ) \in \bZ^3$. We assume that $\bb $ extends to a $\bZ$-basis $\{ \ba, \bc, \bb \}$ of $\bZ^3$, i.e., there is no $\bb' \in \bZ^3$ and $k \in \bZ \setminus \{ 1,0,-1\}$ such that $k\bb'= \bb$. 
Let $T_\bb := (\ba \ \bc \ \bb) \in GL_3(\bZ)$, i.e., $T_\bb \bx = \ba$, $T_\bb \by = \bc$ and $T_\bb \bz = \bb$. 
Then the dislocated lattice is defined as the discrete subset $T_\bb \cdot \bX_{\bz} \subset \bR^3$.  
We define the unitary $\Psi_\bb \colon \ell^2(\bX_\bz) \to \ell^2(\bX_\bb)$ as $\Psi_\bb(\delta_\bv) = \delta_{T_\bb \bv}$. We define the C*-algebra  $\cA_\bb$ as $\Ad (\Psi_\bb)(\cA_\bz)$, which is generated by three unitaries $\widetilde{S}_\ba:=\Psi_\bb \widetilde{S}_\bx \Psi_\bb^*$, $\widetilde{S}_{\bc}:= \Psi_\bb \widetilde{S}_\by \Psi_\bb^*$, $\widetilde{S}_\bb:=\Psi_\bb \widetilde{S} _\bz \Psi_\bb^*$ and $\Ad(\Psi_\bb)(\bK \otimes C^*_r(\bZ))$. 

Then we have a commutative diagram of exact sequences
\[
\xymatrix{
0 \ar[r] & \bK \otimes C^*_r(\bZ)  \ar[r] \ar[d]^{\Ad (\Psi_\bb)} & \cA_\bz \ar[r] \ar[d]^{\Ad (\Psi_\bb)} \ar[r] & C^*_r(\bZ^3) \ar[d]^{T_\bb} \ar[r] & 0\\ 
0 \ar[r] &\bK \otimes C^*_r(\bZ)  \ar[r]  & \cA_\bb \ar[r]  \ar[r] & C^r(\bZ^3) \ar[r] & 0\\
}
\]
which shows that the boundary map in K-theory with respect to the second raw is identified with the one given in \cref{thm:main} through $T_\bb $. 

\subsection{Translation invariance in the $xy$-direction}\label{rmk:transl}
In \cref{section:formulation,section:proof},  we assume that the Hamiltonian $H$ is translation invariant away from the line defect. Indeed, the proof of \cref{thm:main} shows that this assumption is not needed. 
Let us only assume that the Hamiltonian $H \in \bB(\ell^2(\bZ^3))$ is translation invariant in the $z$-direction. 
Let $C^*_u(|\bZ^3|)^{\bZ}$ denote the closure of the set of finite propagation operators on $\ell^2(\bZ^3)$ which is translation invariant in the $z$-direction (this C*-algebra is called the invariant uniform Roe algebra). Here $|\bZ^3|$ stands for the lattice $\bZ^3$ regarded as a metric space. 
This is a subalgebra of the Roe algebra $C^*(|\bZ^3|)$. 
In \cite{kubotaControlledTopologicalPhases2017}, the (possibly not translation invariant) topological phases is classified by the K-theory of the Roe algebra. (More precisely, in \cite{kubotaControlledTopologicalPhases2017} two kinds of classifications are considered; K-theory of the uniform Roe algebra and that of the Roe algebra. The latter classification is more rough.)  The invariant Roe algebra $C^*(|\bZ^3|)^{\bZ}$ is isomorphic to $C^*(|\bZ^2|) \otimes C^*_r(\bZ)$, where the right tensor component is generated by the unitary $S_z$. By the K\"{u}nneth theorem, its K-theory is isomorphic to 
\[\K(C^*(|\bZ^3|)^{\bZ}) \cong \K_{*-2}  \oplus \K_{*-3} = \bZ \beta_{xy} \oplus \bZ \beta_{xyz}, \]
where $\beta_{xy}$ and $\beta_{xyz}$ are the image of the corresponding element with respect to the inclusion $C^*_r(\bZ^3) \subset C^*(|\bZ^3|)^\bZ$.
Now the $\ast$-homomorphism $s$ defined in \cref{lem:sigma} extends to
\[\tilde{s} \colon C^*(|\bZ^3|)^{\bZ} \cong C^*(|\bZ^2|) \otimes C^*_r(\bZ) \to C^*(\widetilde{X}_{\bz,1})^\bZ / C^*(\partial \widetilde{X}_{\bz,1} \subset \widetilde{X}_{\bz,1})^\bZ \]
as $\tilde{s}(T \otimes S_z^k):=s(T) \otimes \widetilde{S}_z^k$. The induced map in K-theory 
\[\K_*(C^*(|\bZ^3|)^{\bZ}) \xrightarrow{s_*} \K_*\bigg( \frac{C^*(\widetilde{X}_{\bz,1})^\bZ}{C^*(\partial \widetilde{X}_{\bz,1} \subset \widetilde{X}_{\bz,1})^\bZ }\bigg) \xrightarrow{\partial} \K_*(C^*(\partial \widetilde{X}_{\bz,1} \subset \widetilde{X}_{\bz,1})^{\bZ})\]
sends $\beta_{xy}$ to $\beta$ and $\beta_{xyz}$ to zero.

\subsection{Bulk-dislocation correspondence of topological insulators}\label{rmk:TI}
Our proof of \cref{thm:main} is given without any use of particularity of the complex K-theory (indeed it relies only on the formal diagram chasing argument and fundamental facts of coarse index theory listed in \cref{rmk:coarse}). Therefore the same proof also work for Real K-theory of C*-algebras. Namely, it is shown that the boundary map $\partial$ of the exact sequence \eqref{eq:ext} sends $\beta_{xy}^{n-2} \in \KR_n(C^*_r(\bZ^3))$ (cf.\ \cref{rmk:real}) to $\beta \in \KR_{n-1}(\bK \otimes C^*_r(\bZ ))$, and other generators to zero. 

This enables us to generalize the bulk-dislocation correspondence for topological insulators with one of the Altland--Zirnbauer 10-fold way symmetry \cite{altlandNonstandardSymmetryClasses1997}. 
This is a class of symmetries generated by some of time-reversal symmetry $T$ (antilinear unitary commuting with $H$), the particle-hole symmetry $C$ (antilinear unitary anticommuting with $H$) and sublattice symmetry $S $ (linear unitary anticommuting with $H$) with the relations $T^2=\pm 1$, $C^2=\pm 1$, $S^2=1$, $S=CT$. 
The set of self-adjoint invertible operators with the symmetry is classified by one of 2 complex and 8 Real K-groups of $C^*_r(\bZ^d)$ (for this, we use Van Daele's description of Real $\bZ_2$-graded C*-algebra K-theory \cite{kellendonkAlgebraicApproachTopological2017} or, equivalently, a twisted equivariant K-theory \cite{kubotaNotesTwistedEquivariant2016}). 

We list in \cref{tab:my_label} the topological classification of $3$-dimensional weak insulators in the $xy$-direction, which is the same as the classification of $2$-dimensional strong topological insulators. 
There exists non-trivial bulk-dislocation correspondence for type A, AI, D, DIII, AII, C topological insulators. 
In the literature of physics, this is already discovered by Teo--Kane~\cite{teoTopologicalDefectsGapless2010}. 

\begin{table}[t]
    \centering
    \begin{tabular}{|r|c|c|c|c|c|c|c|c|c|c|c|c|}\hline
        Type &  A & AIII &AI & AIII & AI & BDI & D & DIII & AII & CII & C &  CI   \\ \hline 
        K-group & $\bZ$ & $0$ & $\bZ$ & $0$ & $0$ & $0 $ & $\bZ$ & $\bZ_2$ & $\bZ_2$ & $0$ & $\bZ$ & $0$ \\ \hline  
    \end{tabular}
    \caption{Strong invariants in dimension $2$}
    \label{tab:my_label}
\end{table}

\subsection{Quantum Hall effect along the dislocation}\label{rmk:IQHE}
Here we discuss a consequence of \cref{thm:main} in the case of type A topological insulator. 
Let $H$ be a $3$-dimensional type A topological insulator. 
Parallel to the case of bulk-edge correspondence of integer quantum Hall effect, the topology of the bulk Hamiltonian $H$ induces the Hall conductivity along the line defect. Following \cite{kellendonkEdgeCurrentChannels2002}*{Equation (24)}, the Hall conductance along the dislocation would be calculated as
\[\sigma_{\mathrm{screw}} =\frac{e^2}{h}  \cdot \Big(  - \lim_{\Delta \to \{ \mu \}} \frac{1}{|\Delta|}  \int_{\bT } \mathop{\mathrm{Tr}}(\widetilde{P}_{\Delta}(k) (\partial _{k_z} \widetilde{H})(k))dk \Big) ,  \]
where $\mathop{\mathrm{Tr}}$ is the (unbounded) trace on the compact operator algebra $\bK$, $\widetilde{P}_\Delta$ is the spectral projection of $\widetilde{H}_\bz$ with respect to the interval $\Delta$, and the integral is taken through the identification $C^*_r(\bZ) \cong C(\bT)$. Note that, since the operators $\widetilde{P}_\Delta$ and $\widetilde{H}$ are $\bZ$-invariant in the $z$-direction, the Fourier transform in the $z$-direction identifies them with the corresponding operator-valued functions on $\bT$. 
As is shown in \cite{kellendonkEdgeCurrentChannels2002}*{Theorem 1}, 
this value coincides with $\partial [\widetilde{H}_\bz]$ through the identification $\K_1(\bK \otimes C^*_r(\bZ)) \cong \bZ$ given by a cyclic $1$-cocycle on $C^*_r(\bZ)$. 

\cref{thm:main} shows that $\sigma_{\mathrm{screw}}$ coincides with the weak Chern number in the $xy$-direction of the bulk Hamiltonian $H$, which is calculated as by the Chern-Weil theory as
\[ \sigma_{\mathrm{bulk}}^{xy} = \frac{e^2}{h} \cdot (2\pi i) \cdot \int_{\bT^2 \times \{0\}} \tr (P[\partial_{k_x}P , \partial_{k_y}P]) dk_xdk_y, \]
where $\tr$ denotes the trace on $\bM_N$ and $P$ denotes the spectral projection of $H$ (regarded as a matrix-valued function on $\bT^3$ through the isomorphism $C^*_r(\bZ^3) \cong C(\bT^3)$) corresponding to the negative eigenvalues. The right hand side is analogous to the TKNN formula in dimension $2$.

\subsection{Relation to the codimension 2 transfer map}\label{rmk:codim2}
The proof of \cref{thm:main} given in this paper is originated from the C*-algebraic codimension 2 transfer map in higher index theory introduced in \cites{kubotaGromovLawsonCodimensionObstruction2020,kubotaCodimensionTransferHigher2021}. 
Let $M$ be a manifold and let $N$ be its codimension $2$ submanifold such that $\pi_1(N) \to \pi_1(M)$ is injective, $\pi_2(N) \to \pi_2(M)$ is surjective, and the normal bundle of $N $ is trivial. Set $\Gamma :=\pi_1(M)$ and $\pi:=\pi_1(N)$. 
In \cite{kubotaGromovLawsonCodimensionObstruction2020}*{Theorem 1.1}, a group homomorphism
\[ \tau_\sigma \colon \K_*(C^*\Gamma  ) \to \K_{*-2}(C^*\pi) \]
is constructed (the notation $\tau_\sigma$ is introduced in \cite{kubotaCodimensionTransferHigher2021}). This homomorphism $\tau_\sigma$ satisfies $\tau_\sigma (\alpha_\Gamma (M)) = \alpha_\pi (N)$. Here, for a closed spin manifold $M$, $\alpha_\Gamma (M)$ denotes the higher index of the Dirac operator on the universal covering $\widetilde{M}$. 
In \cite{kubotaGromovLawsonCodimensionObstruction2020}, a C*-algebra extension
\[ 0 \to \bK_{C^*(\pi \times \bZ) } \to A \to C^*\Gamma \to 0  \]
is constructed, and $\tau_\sigma$ is defined to be the boundary map in K-theory. 

In the simplest case, when $M =\bT^2$ and $N=\pt$, this extension is the same thing as \eqref{eq:ext2d}. 
The fact $\tau_\sigma (\alpha_\Gamma (M)) = \alpha_\pi (N)$ corresponds to \cref{thm:main}. 
The proof in this paper, particularly the construction of a lifting map $s$ in \cref{lem:sigma}, is a special case of the one constructed in Lemma 3.14, Proposition 4.3, and 4.7 of \cite{kubotaCodimensionTransferHigher2021}.

\appendix 
\section{Proof of the lifting lemma}\label{section:appendix}
\cref{lem:sigma} is an essential ingredient of the proof of \cref{thm:main}. Its proof is given in \cite{kubotaCodimensionTransferHigher2021} in a more general setting concerned with codimension $2$ inclusions of manifolds (cf. \cref{rmk:codim2}). 
In this appendix, we restate the proof given in \cite{kubotaCodimensionTransferHigher2021} in a way that is specific to the setting we need, i.e., the coarse index theory of the helical surface.

For a bounded operator $K$ with propagation less than $R>0$, we define its lift $\widetilde{K}_R$ as \eqref{eq:lift}. 
If $K$ is a locally Hilbert-Schmidt operator (i.e., $Kf$ and $fK$ are Hilbert--Schmidt for any $f \in C_c(X)$), then the kernel function $k$ exists and is a Borel function on $X \times X$. 
Even if $K$ is not locally Hilbert-Schmidt, the kernel function $k$ makes sense as a distribution on $X \times X$ supported near the diagonal. Its lift \eqref{eq:kernel} is also defined, and hence $\widetilde{K}_R$ is well-defined as a linear map $\sD(X) \to \sD'(X)$, where $\sD(X)$ denotes the Fr\'{e}chet space of compactly supported test functions. Note that the image of $\widetilde{K}_R$ is included to $L^2(X)$.

\begin{para}
We show that $\widetilde{K}_R$ is a bounded operator with respect to the $L^2$-norms on the domain and the range. 
For $\bv \in \bZ^2$, let $P_{\bv}$ denote the projection onto the $L^2$-space of $(\bv + [0,1) \times [0,1))$. Similarly, for $\tilde{\bv} \in \bX_{\bz}$, let $P_{\tilde{\bv}}$ denote the projection onto the $L^2$-space of the connected component of $\pi^*(\pi(\tilde{\bv}) + [0,1) \times [0,1))$ containing $\tilde{\bv}$. 
Now $\widetilde{K}_R$ is decomposed into a sum
\[\widetilde{K}_R = \sum_{\br \in \bZ^2} \Big( \sum_{\tilde{\bv}\in \bX_{\bz}} P_{[\tilde{\bv} + \br]} \widetilde{K}_R P_{\tilde{\bv}} \Big), \]
which is finite with respect to $\br$, and the norm of each summand is bounded as 
\[\Big\| \sum_{\tilde{\bv} \in \bX_\bz} P_{[\tilde{\bv} + \br]} \widetilde{K}_R P_{\tilde{\bv}}\Big\| \leq \sup_{\tilde{\bv} \in \bX_{\bz}} \| P_{[\tilde{\bv} +\br]} \widetilde{K}_R P_{\tilde{\bv}} \| = \sup_{\tilde{\bv} \in \bX_{\bz}} \| P_{\pi(\tilde{\bv}) +\br} K P_{\pi(\tilde{\bv})} \| \leq \| K\|. \]
\end{para}

\begin{para}
We observe that, if $K$ is locally compact (resp.~pseudo-local), then $\widetilde{K}_R$ is also locally compact (resp.~pseudo-local). We firstly notice that local compactness and pseudo-locality of an operator is a local condition, i.e., it is enough to check that any $f \in C_c(X)$ supported in an $\varepsilon $-ball satisfies $f\widetilde{K}_R, \widetilde{K}_Rf \in \bK$ (resp.~$[\widetilde{K}_R,f] \in \bK$). 
To see this, notice that $\widetilde{K}_Rf$ and $f \widetilde{K}_R$ are supported in an $(\varepsilon +R)$-ball in $\widetilde{X}_{\bz,1}$, which is identified with an $(\varepsilon +R)$-ball in $X$ through the covering map. 
This identification gives rise to a partial isometry of $L^2$-spaces, which identifies $\widetilde{K}_Rf$ and $f \widetilde{K}_R$ with $Kf$ and $fK$ respectively.  
\end{para}

\begin{para}\label{para:A3}
We show that the map $K \mapsto \widetilde{K}_R$ is multiplicative modulo the ideal $C^*(\partial \widetilde{X}_{\bz, 1} \subset \widetilde{X}_{\bz,1})^\bZ $. 
For locally Hilbert--Schmidt operator $K, L \in \bB(L^2(X))$ with $\mathrm{Prop}(K)<R$ and $\mathrm{Prop}(L)<S$, the composition $\widetilde{K}_R\widetilde{L}_S$ is given by convolution with 
\begin{align*}
t(\tilde{x}, \tilde{z}) = \int_{\tilde{y} \in \widetilde{X}_{\bz, R}} \widetilde{k}_R(\tilde{x}, \tilde{y}) \widetilde{l}_S (\tilde{y}, \tilde{z}) d\tilde{y}.
\end{align*}
If $\tilde{x},\tilde{z}$ and the boundary $\partial \widetilde{X}_{\bz, 1}$ are separated by a distance $S + R$, then
\begin{align*}
    t(\tilde{x}, \tilde{z})= \begin{cases} \int_{y \in X} k(\pi(\tilde{x}), y)l(y,\pi(\tilde{z}))dy & \text{ if $d(\tilde{x}, \tilde{z}) <R+S$,} \\ 0 & \text{ otherwise.} \end{cases}
\end{align*}
This shows that $\widetilde{K}_R \widetilde{L}_S$ and $\widetilde{KL}_{R+S}$ coincides modulo $C^*(\partial \widetilde{X}_{\bz, 1} \subset \widetilde{X}_{\bz,1})^\bZ $.
\end{para}

Now we get a $\ast$-homomorphism 
\[s \colon \bC_{\mathrm{HS}} [X] \to C^*(\widetilde{X}_{\bz,1})^\bZ / C^*(\partial \widetilde{X}_{\bz,1} \subset \widetilde{X}_{\bz,1})^\bZ,\]
where $\bC_{\mathrm{HS}}[X]$ denotes the set of locally Hilbert--Schmidt operators with finite propagation. This extends to a $\ast$-homomorphism from $C^*(X)$ since the operator norm on $\bC_{\mathrm{HS}}[X]$ is the largest norm satisfying the C*-condition (this is a consequence of the amenability of the group $\bZ^2$).

\begin{para}
We extend $s$ to a $\ast$-homomorphism between pseudo-local coarse C*-algebras. This part is completely the same as \cite{kubotaCodimensionTransferHigher2021}*{Proposition 4.3}. We repeat the proof just for self-consistency of this appendix.  
As is shown in \cite{roeIndexTheoryCoarse1996}*{Lemma 5.8}, any operator $T \in D^*(X)$ is decomposed as $T = T_0 + T_1$, where $\mathrm{Prop}(T) < 1/2$ and $T^1 \in C^*(X)$. Let $\Pi$ denote the projection onto $L^2(\widetilde{X}_{\bz,1}) \subset L^2(\widetilde{X}_{\bz, 1/2})$. Set
\[s(T):= \Pi (\widetilde{T_0})_{1/2} \Pi + s(T_1) \in D^*(\widetilde{X}_{\bz,1})^\bZ / D^*(\partial \widetilde{X}_{\bz,1} \subset \widetilde{X}_{\bz,1}). \]
This is well-defined independent of the choice of the decomposition $T = T_0 + T_1$. Indeed, if we have another decomposition $T=T_0'+T_1'$, then 
\[ (\Pi (\widetilde{T_0})_{1/2} \Pi + s(T_1)) - (\Pi (\widetilde{T_0'})_{1/2} \Pi + s(T_1')) = \Pi ((\widetilde{T_0})_{1/2}- (\widetilde{T_0'})_{1/2} \Pi + s(T_1 - T'_1).  \]
Since $T_0-T_0' = -(T_1 - T_1')$ is contained in $C^*(\widetilde{X}_{\bz,1})$ and has propagation less than $1/2$, the right hand side lies in the ideal $C^*(\partial \widetilde{X}_{\bz,1} \subset \widetilde{X}_{\bz,1})$. 

We also show that this $s$ is multiplicative. For $T,S \in D^*(X)$, we choose decompositions $T=T_0+T_1$ and $S = S_0+S_1$ as $\mathrm{Prop}(T_0)$ and $\mathrm{Prop}(S_0)$ are less than $1/4$. Then we have
\begin{align*}
    s(T)s(S) - s(TS) =& (\Pi (\widetilde{T_0})_{1/2} \Pi(\widetilde{S_0})_{1/2}\Pi  - \Pi (\widetilde{T_0})_{1/2} (\widetilde{S_0})_{1/2} \Pi ) \\
    & + (\Pi (\widetilde{T_0})_{1/2} \Pi s(S_1) - s(T_0S_1)) \\
    & + (s(T_1) \Pi (\widetilde{S_0})_{1/2} \Pi - s(T_1S_0)) \\
    & + (s(T_1)s(S_1) - s(T_1S_1)).
\end{align*}   
By \cref{para:A3}, the second, third and forth components are locally compact and supported in a neighborhood of $\partial \widetilde{X}_{\bz,1}$, i.e., are contained in $C^*(\partial \widetilde{X}_{\bz,1} \subset \widetilde{X}_{\bz,1})^\bZ$. 
Moreover, the first term is contained in $D^*(\partial \widetilde{X}_{\bz,1} \subset \widetilde{X}_{\bz,1})$. Indeed, for any $f \in C_c(\widetilde{X}_{\bz,1})$ vanishing at the boundary, the composition
\begin{align*}
    (\Pi (\widetilde{T_0})_{1/2} \Pi(\widetilde{S_0})_{1/2}\Pi  - \Pi (\widetilde{T_0})_{1/2} (\widetilde{S_0})_{1/2} \Pi ) f & = \Pi \widetilde{T}{}^0_{1/2} (1-\Pi) \widetilde{S}{}^0_{1/2} f  \\
    & = \Pi \widetilde{T}{}^0_{1/2} (1-\Pi) [\widetilde{S}{}^0_{1/2}, f] 
\end{align*}
is in $\bK(L^2(\widetilde{X}_{\bz,1}))$ (here we use $f\Pi = f$). This shows the multiplicativity of $s$. 
\end{para}

\begin{para}
Finally we show that $s$ sends the fundamental class $[X]$ to $[\widetilde{X}_{\bz,1} , \partial \widetilde{X}_{\bz, 1}]$. 
We start with a definition of these fundamental classes. 
Let $Z$ be a complete $2$-dimensional Riemannian manifold obtained by attaching to  $X_{\bz, 1}$ an infinite cylinder of the boundary. 
We choose $0$-th order pseudo-differential operators $F \in \bB(L^2(X,\bC^2))$ and $F_Z \in \bB(L^2(Z,\bC^2))$ whose principal symbols are the same as that of Dirac operators on $X$ and $Z$ respectively (note that the spinor bundle of $X = \bR^2$ is $\bC^2$). 
We may choose them as $\mathrm{Prop}(F) < 1/2$ and $\mathrm{Prop}(F_Z) < 1/2$ (\cite{lawsonjr.SpinGeometry1989}*{Corollary III.3.7}). 
Then both $F \in Q^*(X)$ and $\Pi \widetilde{F} \Pi \in Q^*(\widetilde{X}_{\bz,1})/Q^*(\partial \widetilde{X}_{\bz,1} \subset \widetilde{X}_{\bz,1})$ are unitaries. The $\K_1$-classes determined by them are denoted by $[X]$ and $[\widetilde{X}_{\bz,1}, \partial \widetilde{X}_{\bz,1}]$ respectively. 

Now, the lift $\widetilde{F}_{1/2}$ restricted to $\widetilde{X}_{\bz,1/2}$ is a $0$-th order pseudo-differential operator whose principal symbol is  $\sigma(\widetilde{F}_{1/2}) = \sigma (F_Z)|_{X_{\bz,1}}$. Hence, for any $f \in C_c(\widetilde{X}_{\bz,1})$ vanishing at the boundary, we have
\[(\Pi \widetilde{F}_{1/2} \Pi - \Pi F_Z \Pi)f = [\widetilde{F}_{1/2} - F_Z, f] + f (\widetilde{F}_{1/2} - F_Z) \in \bK(L^2(\widetilde{X}_{\bz,1})). \]
This shows that $\Pi \widetilde{F}_{1/2}\Pi = \Pi F_Z \Pi$ modulo $D^*(\partial \widetilde{X}_{\bz,1} \subset \widetilde{X}_{\bz,1})^\bZ + C^*(\widetilde{X}_{\bz,1})^\bZ$, in other words, their images coincide in $Q^*(\widetilde{X}_{\bz,1})^\bZ /Q^*(\partial \widetilde{X}_{\bz,1} \subset \widetilde{X}_{\bz,1})^\bZ$.
This shows $s_*[X] = [\widetilde{X}_{\bz,1}, \partial \widetilde{X}_{\bz,1}]$.
\end{para}

\bibliographystyle{alpha}
\bibliography{ref.bib}

\end{document}